\documentclass[12pt,a4paper]{article}
\usepackage[a4paper,total={5.50in, 8.0in}]{geometry}
\usepackage{graphicx,hyperref}
\usepackage[round]{natbib}
\usepackage{amsmath,amssymb,color}
\usepackage[linewidth=2pt,]{mdframed}
\usepackage{xcolor}

\usepackage{amsthm}
\newtheorem{proposition}{Proposition}

\begin{document}

\begin{center}
\begin{LARGE}
\textbf{Note on the robustification of the Student $t$-test statistic
using the median and the median absolute deviation}
\end{LARGE}

\medskip
\begin{large}
Chanseok Park 
\end{large}

\medskip
\begin{itshape}
Applied Statistics Laboratory \\
Department of Industrial Engineering \\
Pusan National University, Busan 46241, Republic of Korea 
\end{itshape}
\end{center}

\begin{abstract}
\noindent
In this note, 
we propose a robustified analogue of the conventional Student $t$-test statistic.
The proposed statistic is easy to implement and thus practically useful.
We also show that it is a pivotal quantity and converges to a standard normal distribution.

\bigskip
\noindent\textbf{\uppercase{Keywords:}} Student $t$-test statistic, 
median, median absolute deviation, pivot, 
robustness, asymptotic normality. 
\end{abstract}

\section{Introduction}
The conventional Student $t$-test statistic
\[ 
T=  \frac{\bar{X}-\mu}{S/\sqrt{n}}, 
\]
is one of the most widely used test statistics in statistical applications. \
The downside of this statistic, however, is that it is 
extremely sensitive to contamination. 
In this note, we introduce a robustified analogue of the conventional Student $t$-test
statistic which is robust to data contamination. 
The proposed statistic also possesses the important property of asymptotic normality.  
We provide the proof of its asymptotic normality in the following section.

\section{Asymptotic properties of the proposed methods}
We propose the test statistic defined as follows: 
\begin{equation} \label{EQ:TeststatmM0}
T_{m} = \frac{ \hat{\mu}_m - \mu }{ \hat{\sigma}_M/\sqrt{n} }. 
\end{equation}
where $\hat{\mu}_m$ is the median estimator and $\hat{\sigma}_M$ is 
the median absolute deviation (MAD) estimator which are defined as
\begin{align*} 
\hat{\mu}_m 
&= \mathop{\mathrm{median}}_{1\le i\le n} X_i 
\intertext{and}
\hat{\sigma}_M 
&= \mathop{\mathrm{median}}_{1\le i\le n} 
  \big| X_i - \mathop{\mathrm{median}}_{1\le i\le n} X_i  \big|,
\end{align*} 
respectively.

Then a natural question arises: Does the proposed statistic, 
$T_m$, converge to the standard normal distribution under $H_0: \mu=\mu_0$? 
In what follows, we answer this question.

\begin{proposition} \label{THM:pivotmM}
Let $X_1,X_2,\ldots,X_n$ be a random sample from a location-scale family with
location $\mu$ and scale $\sigma$.
Then the following statistic is a pivotal quantity: 
\begin{equation} 
\frac{\displaystyle\mathop{\mathrm{median}}_{1\le i\le n}X_i-\mu}%
{\displaystyle\mathop\mathrm{MAD}_{1\le i\le n}X_i}
= \frac{\displaystyle\mathop{\mathrm{median}}_{1\le i\le n}X_i-\mu}%
{\displaystyle\mathop{\mathrm{median}}_{1\le i\le n} 
  \big| X_i - \mathop{\mathrm{median}}_{1\le i\le n} X_i  \big|}
\label{EQ:pivotmM}
\end{equation}
\end{proposition}
\begin{proof}
Using Theorem 3.5.6 of \cite{Casella/Berger:2002}, 
we can write 
\[
X_i = \sigma Z_i + \mu, 
\]
where the distribution of $Z_i$ has the distribution of $Z_i$ with $\mu=0$ and $\sigma=1$.
Note that the distribution of $Z_i$ does not include any unknown parameters.
Next, it is easily seen that  
\begin{equation}
\mathop{\mathrm{median}}_{1\le i\le n}X_i 
   =\sigma \cdot \mathop{\mathrm{median}}_{1\le i\le n}Z_i + \mu 
\label{EQ:mumedian}
\end{equation}
and
\begin{align}
\mathop{\mathrm{median}}_{1\le i\le n} 
  \big| X_i - \mathop{\mathrm{median}}_{1\le i\le n} X_i  \big|  
&=\mathop{\mathrm{median}}_{1\le i\le n} 
  \big|\sigma Z_i+\mu - 
  (\sigma \cdot \mathop{\mathrm{median}}_{1\le i\le n}Z_i+\mu)\big|\notag\\
&= \sigma \mathop{\mathrm{median}}_{1\le i\le n} 
   \big| Z_i - \mathop{\mathrm{median}}_{1\le i\le n}Z_i \big| .
\label{EQ:sigmaMAD} 
\end{align}
It is immediate upon substituting (\ref{EQ:mumedian}) 
and (\ref{EQ:sigmaMAD}) into (\ref{EQ:pivotmM}) that we have 
\[
\frac{\displaystyle\mathop{\mathrm{median}}_{1\le i\le n}Z_i}%
{\displaystyle\mathop{\mathrm{median}}_{1\le i\le n} 
  \big| Z_i - \mathop{\mathrm{median}}_{1\le i\le n} Z_i  \big|} .
\]
which shows that the expression in (\ref{EQ:pivotmM}) 
does not include any unknown parameters. This completes the proof.
\end{proof}

\begin{proposition} \label{THM:TeststatmM}
Let $X_1,X_2,\ldots,X_n$ be a random sample from a normal distribution 
with $\mu$ and variance $\sigma^2$. 
Then we have 
\begin{equation} \label{EQ:TeststatmM1}
\sqrt{\frac{2n}{\pi}} {\Phi^{-1}\Big(\frac{3}{4}\Big)} \cdot 
\frac{\displaystyle\mathop{\mathrm{median}}_{1\le i\le n}X_i-\mu}%
{\displaystyle\mathop{\mathrm{median}}_{1\le i\le n} 
  \big| X_i - \mathop{\mathrm{median}}_{1\le i\le n} X_i  \big|}
\stackrel{d}{\longrightarrow} N(0,1).
\end{equation}
\end{proposition}
\begin{proof}
The normal distribution is a part of the location-scale family 
of distributions.
Thus, given Proposition~\ref{THM:pivotmM}, it is sufficient 
to show that 
\[
\sqrt{\frac{2n}{\pi}} {\Phi^{-1}\Big(\frac{3}{4}\Big)} \cdot 
\frac{\displaystyle\mathop{\mathrm{median}}_{1\le i\le n}Z_i}%
{\displaystyle\mathop{\mathrm{median}}_{1\le i\le n} 
  \big| Z_i - \mathop{\mathrm{median}}_{1\le i\le n} Z_i  \big|} 
\stackrel{d}{\longrightarrow} N(0,1),
\]
where $\Phi^{-1}(\cdot)$ is the quantile function of the standard normal distribution.

From Example 2.4.9 of \cite{Lehmann:1999}, we have
\[
\sqrt{n}(\hat{\mu}_m-\mu) \stackrel{d}{\longrightarrow} 
N\bigg(0, \frac{1}{4 f^2(0)}  \bigg), 
\]
where $\hat{\mu}_m = \mathop{\mathrm{median}}_{1\le i\le n}X_i$ again 
and  $f(\cdot)$ is the probability density function of the random variable $X_i$.

Since the distribution of $Z_i$ is the standard normal,  
we have $f(0)=1/\sqrt{2\pi}$.
Using Example 2.4.9 of \cite{Lehmann:1999}, we have 
\[
\sqrt{n} \mathop{\mathrm{median}}_{1 \le i \le n}Z_i 
 \stackrel{d}{\longrightarrow} 
N\bigg(0, \frac{\pi}{2}  \bigg).
\]
Therefore, it clearly follows that 
\[
\sqrt{\frac{2n}{\pi}} \mathop{\mathrm{median}}_{1 \le i \le n}Z_i 
 \stackrel{d}{\longrightarrow} N\big(0,1 \big).
\]

Next, we consider the rescaled MAD estimator 
\[
\frac{\displaystyle\mathop{\mathrm{median}}_{1\le i\le n}
\big| Z_i - \mathop{\mathrm{median}}_{1\le i\le n} Z_i \big|}%
{\Phi^{-1}\big({3}/{4}\big)}.
\]
In Section 5.1 of \cite{Huber/Ronchetti:2009}, it is proven
that the rescaled MAD estimator is Fisher-consistent under the normal distribution.
Clearly, since $Z_i$ are independent and identically distributed
standard normal random variables,  the rescaled MAD estimator is Fisher-consistent 
for $\sigma=1$. 

Now, it is well known that Fisher-consistency implies weak consistency.
For more details, see Section 9.2 (ii) of \cite{Cox/Hinkley:1974}. 
Therefore, combining the weak 
consistency result and Lemma 2.8  (Slutsky) in \cite{Vaart:1998}, we then have: \[
\sqrt{\frac{2n}{\pi}} {\Phi^{-1}\Big(\frac{3}{4}\Big)} \cdot 
\frac{\displaystyle\mathop{\mathrm{median}}_{1\le i\le n}Z_i }%
{ \displaystyle\mathop{\mathrm{median}}_{1\le i\le n}
\big| Z_i - \mathop{\mathrm{median}}_{1\le i\le n} Z_i \big|}
 \stackrel{d}{\longrightarrow} N\big(0, 1 \big),
\]
which completes the proof.
\end{proof}
Thus, although the test statistic $T_m$ in (\ref{EQ:TeststatmM0}) 
does not converge to $N(0,1)$,
the scaled statistic in Proposition~\ref{THM:TeststatmM} does. Specifically, 
we have obtained the following convergence result:
\[
  \sqrt{\frac{2n}{\pi}} \cdot \Phi^{-1}(3/4) \cdot T_m 
 \stackrel{d}{\longrightarrow}  N(0,1). 
\]

\section*{Acknowledgment}
This note was supported  in part 
by the National Research Foundation of Korea (NRF) grant
funded by the Korea government (No. NRF-2017R1A2B4004169).

\bibliographystyle{apalike}
 \bibliography{$HOME/MyFiles/library/TeX/bib/STAT,%
 $HOME/MyFiles/library/TeX/bib/CP,%
 $HOME/MyFiles/library/TeX/bib/ENG}
\end{document}